\documentclass{toc}

\usepackage{xspace}

\theoremstyle{plain}
\newtheorem{observation}[theorem]{Observation}
\newtheorem{open}[theorem]{Open Question}
\newtheorem*{theorem*}{Theorem}

\newcommand{\definedWord}[1]{\emph{#1}}
\newcommand{\rothvoss}{Rothvoss\xspace}

\DeclareMathOperator{\perm}{perm}
\DeclareMathOperator{\Sat}{Sat}
\DeclareMathOperator{\Clow}{Clow}
\DeclareMathOperator{\Clique}{Clique}
\DeclareMathOperator{\Cut}{Cut}
\DeclareMathOperator{\New}{Newt}
\DeclareMathOperator{\xc}{xc} %
\DeclareMathOperator{\compl}{c} %
\DeclareMathOperator{\poly}{poly} %

\newcommand{\HC}{\operatorname{HC}}

\newcommand{\B}{\mathbb{B}}
\newcommand{\F}{\mathbb{F}}

\newcommand{\sharpP}{\cclass{\#P}}
\newcommand{\VP}{\cclass{VP}}
\newcommand{\VNP}{\cclass{VNP}}

\tocdetails{%
volume=13, number=18, year=2017, firstpage=1,
received={November 10, 2015},
revised={July 27, 2016},
published={December 22, 2017},
doi={10.4086/toc.2017.v013a018},
title={Monotone Projection Lower Bounds from Extended Formulation Lower Bounds},
author={Joshua A.\ Grochow},
plaintextauthor={Joshua A. Grochow},
acmclassification={F.1.3, F.2.1, G.1.6},
amsclassification={68Q15, 68Q17, 90C05, 15A15, 05C70},
keywords={lower bounds, algebraic circuit complexity, extended formulations of polytopes, %
Newton polytope, monotone formula, monotone circuit, projection, permanent, VNP, matching}
}

\begin{document}

\begin{frontmatter}[classification=text]
\author[grochow]{Joshua A. Grochow\thanks{The author was supported during this work by an Omidyar Fellowship from the Santa Fe Institute.}}

\begin{abstract} %
In this short note, we reduce lower bounds on monotone projections of polynomials to lower bounds on extended formulations of polytopes. %
Applying our reduction to the seminal extended formulation lower bounds of Fiorini, Massar, Pokutta, Tiwari, \& de Wolf (STOC 2012; J. ACM, 2015) and
\rothvoss
(STOC 2014; J. ACM, 2017), we obtain the following interesting consequences. %
\begin{enumerate}
\item The Hamiltonian Cycle polynomial is not a monotone
subexponential-size %
projection of the permanent; this both rules out a natural attempt at a monotone lower bound on the Boolean permanent, and shows that the permanent is \emph{not} complete for non-negative polynomials in $\cclass{VNP}_{\R}$ under monotone p-projections.
\item The cut polynomials and the perfect matching polynomial (or ``unsigned Pfaffian'') are not monotone p-projections of the permanent. The latter, over the Boolean and-or semi-ring, rules out monotone reductions in one of the natural approaches to reducing perfect matchings in general graphs to perfect matchings in bipartite graphs. 
\end{enumerate}
As the permanent is universal for monotone formulas, these results also imply exponential lower bounds on the monotone formula size and monotone circuit size of these polynomials.

\end{abstract}

%
%
%
%
%
%
%
%
%
%

%

\end{frontmatter}

\section{Introduction} \label{sec:intro}
The permanent
\[
\perm_n(X) = \sum_{\pi \in S_n} x_{1,\pi(1)} x_{2,\pi(2)} \dotsb x_{n,\pi(n)}
\]
(where $S_n$ denotes the symmetric group of all permutations of $\{1,\dotsc,n\}$) has long fascinated combinatorists~\cite{minc,vanLintWilson,muirMetzler}, %
more recently physicists~\cite{permQuantum,linearOptics}, and, %
since Valiant's seminal paper~\cite{valiant}, has also been a key object of study in computational complexity. Despite its beauty, the permanent has some computational quirks: in particular, although the permanent of integer matrices is $\sharpP$-complete and the permanent is $\VNP$-complete in characteristic zero, the permanent \emph{mod 2} is the same as the determinant, and hence can easily be computed. In fact, computing the permanent mod $2^k$ is easy for any $k$ \cite{valiant}, though the proof is more involved. 
Modulo any odd number $n$, the permanent of integer matrices is $\cclass{Mod}_n\cclass{P}$-complete~\cite{valiant}.  %

In contrast, the seemingly similar Hamiltonian Cycle polynomial, 
\[
\HC_n(X) = \sum_{\text{$n$-cycles } \sigma} x_{1,\sigma(1)} x_{2,\sigma(2)} \dotsb x_{n,\sigma(n)}\,,
\]
where the sum is only over $n$-cycles rather than over all permutations, does not have these quirks: the %
Hamiltonian Cycle polynomial is $\VNP$-complete over any ring $R$~\cite{valiant2} and $\cclass{Mod}_n\cclass{P}$-complete for all $n$ (that is, counting Hamiltonian cycles is complete for these Boolean counting classes).

Jukna~\cite{juknaQ} observed that, over the Boolean semi-ring, if the Hamiltonian Cycle polynomial were a monotone p-projection of the permanent, there would be a $2^{n^{\Omega(1)}}$ lower bound on monotone circuits computing the permanent, a lower bound that still remains 
open. (The current record is still Razborov's 
$n^{\Omega(\log n)}$~\cite{razborov}.) 
Even over the real numbers, such a monotone p-projection would give an alternative proof of a $2^{n^{\Omega(1)}}$ lower bound on the
permanent.  (Jerrum and Snir~\cite{jerrumSnir} already showed the
permanent requires monotone circuits of size $2^{\Omega(n)}$ over
$\R$ and over the tropical $(\min,+)$ semi-ring.)
Here, by building on Fiorini \emph{et al.}'s~\cite{EF} and
\rothvoss's~\cite{rothvoss}
extended formulation lower bound for the TSP polytope, %
 we show that no such monotone reduction exists---over $\R$, nor over the tropical semi-ring, nor over the Boolean semi-ring---by connecting monotone p-projections to extended formulations of polytopes. %

In the past five %
years, there has been exciting progress on extended formulations of
polytopes,  %
which we leverage by using our new connection. Indeed, with this connection in hand, one immediately gets a monotone projection lower bound from essentially \emph{any} lower bound on extended formulations. An \emph{extended formulation} of a polytope $P$ %
is another %
polytope in a higher-dimensional space that projects down onto $P$ %
by an affine linear map. Since linear programming can be solved in polynomial time, and optima of linear programs are preserved by affine linear projections, solving the LP on the extended formulation allows one to solve the LP on the original polytope. Thus, if a polytope $P$ has an extended formulation that is in some sense ``small,'' one can solve LP optimization problems over $P$ by instead solving them over its smaller extended formulation; if the extended formulation is small enough, this would yield a polynomial-time algorithm. To show that such an extended formulation is not small, it suffices to prove a lower bound on its number of facets. In 1988, Yannakakis~\cite{yan} ruled out a large and natural class of extended formulations of the TSP polytope---so-called \emph{symmetric} extended formulations---thus showing that a certain natural class of algorithms for an $\cclass{NP}$-complete problem indeed did not solve it in polynomial time, and ruling out several attempted proofs that $\cclass{P} = \cclass{NP}$. But for more than 20 years, it was an open question of how to remove the condition of symmetry from Yannakakis's result. In a landmark result, Fiorini, Massar, Pokutta, Tiwary, and de Wolf~\cite{EF} achieved this, by showing an exponential lower bound on \emph{arbitrary} extended formulations of several different polytopes. We will use their lower bound on the cut polytope below.
\rothvoss
\cite{rothvoss} then showed such lower bounds on several other polytopes; we use his lower bound on the TSP polytope, improving that of Fiorini \emph{et al.} \cite{EF}, to get the result above.

We use the same connection to extended formulations, now in combination with
\rothvoss's
lower bound on extended formulations of the perfect matching polytope~\cite{rothvoss}, to show that the perfect matching polynomial or ``unsigned Pfaffian''
\[
\frac{1}{2^n n!} \sum_{\pi \in S_{2n}} \prod_{i=1}^{n} x_{\pi(2i-1), \pi(2i)} = \sum_{\substack{\pi \in S_{2n} \\ \pi(1) < \pi(3) < \dotsb < \pi(2n-1) \\ \pi(2k-1) < \pi(2k) \; \forall k}} \prod_{i=1}^{n} x_{\pi(2i-1), \pi(2i)}
\]
is not a monotone p-projection of the permanent. 
As the perfect matching polynomial counts perfect matchings in a general graph, and the permanent counts perfect matchings in a bipartite graph, 
it is interesting to consider this result in the context of the difference in complication between algorithms for finding %
perfect matchings in bipartite graphs (\eg, \cite{MVV, FGT}) and those for finding %
perfect matchings in general graphs (\eg, \cite{edmonds, MV, vazirani, ST}).

\begin{remark}[On the Boolean semi-ring]
Our results also hold for \emph{formal polynomials} over the Boolean semi-ring $\B = (\{0,1\}, \vee, \wedge)$. Over the Boolean semi-ring, the permanent is the indicator function of the existence of a perfect matching in a bipartite graph, and the unsigned Pfaffian is the indicator function of a perfect matching in a general graph. 
However, over $\B$, each function is represented by more than one formal polynomial, and we do not 
know how to extend our results to the setting of \emph{functions} over $\B$. See \expref{Section}{sec:open} for details and specific questions.
\end{remark}

We also use the same technique, this time in combination with the lower bound of Fiorini \emph{et al.} on extended formulations of the cut polytope~\cite{EF}, to show that the cut polynomials
\[
  \Cut^q = \sum_{A \subseteq [n]} \prod_{i \in A, j \notin A} x_{ij}^{q-1}
\]
are not monotone p-projections of the permanent. Perhaps the main
complexity-theoretic interest in the cut polynomials is that $\Cut^q$
\emph{over the finite field $\F_q$} was (until recently~\cite{MS}) the
only known example of a natural polynomial that is neither 
expected to be in $\VP_{\F_q}$ nor to be $\VNP_{\F_q}$-complete. 
(Indeed, either its membership in $\VP_{\F_q}$ \emph{or} its $\VNP_{\F_q}$-completeness would imply the collapse of the polynomial-time hierarchy~\cite{burgisser}, contradicting a standard complexity-theoretic hypothesis.)
There B\"{u}rgisser also showed that if %
$\VP_{\F_q} \neq \VNP_{\F_q}$ then such polynomials of intermediate
complexity must exist. In that paper, B\"{u}rgisser asked whether the cut %
polynomials, considered as polynomials \emph{over the rationals}, were
$\VNP_{\Q}$-complete.\footnote{$\Cut^2$ was subsequently shown to be
  $\VNP_{\Q}$-complete under circuit reductions~\cite{cut}; its
  completeness under projections remains open.}  Although our results
don't touch on this question, these previous results motivate the
study of these polynomials over $\Q$.

Because the permanent is universal for monotone formulas, our lower bounds also imply exponential lower bounds on the monotone algebraic formula size---and, by balancing algebraic circuits, monotone algebraic circuit size---of these polynomials; see \expref{Section}{sec:formula}.

Finally, we note that our results shed a little more light on the intricacy %
 of the known $\VNP$-completeness proofs for the permanent~\cite{valiant,benDorHalevi,aaronson}. %
Namely, prior to our result, the fact that the permanent is not hard modulo $2$ already implied that any completeness result must use $2$ in a ``bad'' way: for example, dividing by 2 somewhere. This is indeed true of %
Valiant's original proof~\cite{valiant},
Ben-Dor \& Halevi's proof~\cite{benDorHalevi}, %
 and Aaronson's quantum linear-optics proof~\cite{aaronson}. %

\begin{remark}
One might hope for a classical analogue of Aaronson's quantum proof, using the characterization of $\cclass{BPP}$ in terms of stochastic matrices as a replacement for the characterization of $\cclass{BQP}$ using unitary matrices. However, our result
indicates
that the most straightforward adaptation of Aaronson's proof from the $\cclass{BQP}$ setting to the $\cclass{BPP}$ setting cannot work, as the use of stochastic matrices in this manner would produce a monotone reduction.
\end{remark}

Our results also imply the necessity of the use of negative numbers in Valiant's  $4 \times 4$ gadget~\cite[p.~195]{valiant} and Ben-Dor and Halevi's $7 \times 7$ gadget~\cite[App.~A]{benDorHalevi}. %
In light of these results, Valiant's $4 \times 4$ gadget may perhaps seem less mysterious than the fact that such a gadget exists that is \emph{only} $4 \times 4$! 

To prove these results, we show that a monotone projection between non-negative polynomials essentially implies that the Newton polytope of one polynomial is an extension of the Newton polytope of the other (\expref{Lemma}{lem:main}), and then apply known lower bounds on the extension complexity of certain polytopes. We hope that the connection between Newton polytopes, monotone projections, and extended formulations finds further use.

\section{Preliminaries} \label{sec:prelim}
A polynomial $f(x_1,\dotsc,x_n)$ is a \definedWord{(simple) projection} of a polynomial $g(y_1, \dotsc, y_m)$ if $f$ can be constructed from $g$ by replacing each $y_i$ with a constant or with some $x_j$. The polynomial $f$ is an \definedWord{affine projection} of $g$ if $f$ can be constructed from $g$ by replacing each $y_i$ with an affine linear function $\pi_i(\vec{x})$. When we say ``projection'' we mean simple projection. Given two families of 
polynomials, $(f_n)$ and $(g_n)$, 
if there is a function $p(n)$ such that $f_n$ is a projection of $g_{p(n)}$ for all sufficiently large $n$, then we say that $(f_n)$ is a projection of $(g_n)$ with \definedWord{blow-up} $p(n)$. If $(f_n)$ is a projection of $(g_n)$ with polynomial blow-up, we say that $(f_n)$ is a \definedWord{p-projection} of $(g_n)$.

Over any subring of $\R$---or more generally any totally ordered semi-ring (see below)---a \definedWord{monotone projection} is a projection in which all constants appearing in the projection are non-negative. %
Monotone p-projection is defined analogously.

To each monomial $x_1^{e_1} \dotsb x_n^{e_n}$ we associate its exponent vector $(e_1, \dotsc, e_n)$, as a point in $\N^n \subseteq \R^n$.
These vectors determine the Newton polytope, defined as follows.
\begin{definition}
The \definedWord{Newton polytope} of a polynomial $f(x_1, \dotsc, x_n)$, denoted $\New(f)$, is the convex hull in $\R^n$ of the exponent vectors of all monomials appearing in $f$ with non-zero coefficient.
\end{definition}

A polytope is \definedWord{integral} if all its vertices have integer coordinates; note that Newton polytopes are always integral. A \definedWord{face} of a polytope $P$ is the intersection of $P$ with a linear space $L$ such that none of the interior points of $P$ lie in $L$.

For a polytope $P$, let $\compl(P)$ denote the ``complexity'' of $P$, as measured by the minimal number of linear inequalities needed to define $P$ (equivalently, the number of faces of $P$ of dimension $\dim P -1$). A polytope $Q \subseteq \R^m$ is an \definedWord{extension} or \definedWord{extended formulation} %
of $P \subseteq \R^n$ if there is an affine linear map $\ell\colon \R^m \to \R^n$ such that $\ell(Q) = P$. The \definedWord{extension complexity} of $P$, denoted $\xc(P)$, is the minimum complexity of any extension of $P$ (of any dimension): $\xc(P) = \min \{ \compl(Q) \mid \text{$Q$ is an extension of $P$}\}$.

The $m$-th \definedWord{cycle cover polytope} (also known as the bipartite perfect matching polytope) is the convex hull in $\R^{m^2}$ of the $\{0,1\}$-indicator functions of the directed cycle covers of the complete directed graph with self-loops on $m$ vertices. The cycle cover polytope is the Newton polytope of the permanent, as each monomial in the permanent corresponds to such a cycle cover. 

A totally ordered semi-ring (we only consider commutative ones here) is a totally ordered set together with two operations, denoted $(R,\leq,\times,+,1,0)$ such that $(R,\times,1)$ and $(R,+,0)$ are both commutative monoids, $\times$ distributes over $+$, $0 \times a=0$ for all $a$, $a+c \leq b + c$ whenever $a \leq b$, and $ac \leq bc$ whenever $a \leq b$ and $0 \leq c$. An element $c$ of a totally ordered semi-ring is non-negative if $0 \leq c$, and is positive if furthermore $c \neq 0$. We will restrict our attention to \emph{non-zero} totally ordered semi-rings; equivalently, we assume $1 \neq 0$.
Note that in a totally ordered semiring, $1+ \dotsb +1$ is never zero.%

The following totally ordered semi-rings are of particular 
interest.   %
\begin{itemize}
\item The real numbers with its usual ordering and
arithmetic %
operations, $(\R, \leq, \times, +)$. %

\item The so-called ``tropical semi-ring'' $(\R, \leq, +, \min)$, which is the real numbers with its usual ordering, where the product is taken to be real addition and the addition operation is taken to be the minimum.

\item The Boolean and-or semi-ring $\B=(\{0,1\}, \leq, \wedge, \vee)$, where $0 \leq 1$. 
\end{itemize}
To get a feel for the latter two semi-rings, note that polynomials over the tropical semi-ring generally compute some optimization problem, and over $\B$ generally compute a decision problem. For example, the Hamiltonian Cycle polynomial over the tropical semi-ring computes the Traveling Salesperson Problem, and over 
$\B$, the indicator function of the existence of a Hamiltonian cycle. Note that over $\R$, if two formal polynomials compute the same function then they must be identical, but this is not true over the tropical 
or Boolean semi-rings.

\section{Main lemma}
\begin{lemma} \label{lem:main}
Let $R$ be a totally ordered semi-ring, and let $f(x_1,\dotsc, x_n)$ and $g(y_1, \dotsc, y_m)$ be polynomials over $R$ with non-negative coefficients. If $f$ is a monotone projection of $g$, then some face of $\New(g)$ is an extension of $\New(f)$. In particular, $\xc(\New(f)) \leq \compl(\New(g))$.
\end{lemma}

\begin{proof}
Under simple projections, a monomial in the
$y_i$  %
maps to some scalar multiple of a monomial in the 
$x_j$  %
(possibly the empty monomial, resulting in a constant term, or possibly the zero multiple, resulting in zero). Let $\pi$ be a monotone projection map, defined on the variables $y_i$, and extended naturally to monomials and terms in the $y_i$. %
(Recall that a \definedWord{term} of a polynomial is a monomial together with its coefficient.) Since each term $t$ of $g$ is a monomial multiplied by a positive coefficient, and since $\pi$ is non-negative, $\pi(t)$ is either zero or a single monomial in the $x_j$ %
 with nonzero coefficient. The former situation can happen only if $t$ contains some variable $y_i$ such that $\pi(y_i)=0$. Let $\ker(\pi)$ denote the set $\{y_i \mid \pi(y_i) = 0\}$. Thus, for every term $t$ of $g$ that is disjoint from $\ker(\pi)$, $\pi(t)$ actually appears in $f$---possibly with a different coefficient, but still non-zero---since no %
 terms can cancel under projection by $\pi$. 

Let $e_1,\dotsc,e_m$ be the coordinate functions on $\R^m$, %
the ambient space of $\New(g)$; that is, $e_i\colon \R^m \to \R$ is projection onto the $i$-th coordinate. %
Let $K$ denote the subspace of $\R^m$ defined by the equations $e_i = 0$ for each $i$ such that $y_i \in \ker(\pi)$. Let $P$ be the intersection of $\New(g)$ with $K$, considered as a polytope in $K$; since all vertices of $\New(g)$ are non-negative, intersecting $\New(g)$ with a coordinate hyperplane, $e_i = 0$, results in a face of $\New(g)$, and thus $P$ is a face of $\New(g)$. Note that $P$ is exactly the convex hull of the exponent vectors of monomials in $g$ that are disjoint from $\ker(\pi)$. In particular, since $\pi$ is multiplicative on monomials, it induces a \emph{linear} map $\ell_\pi$ from $K$ to $\R^n$ (the ambient space of $\New(f)$). By the previous paragraph, the exponent vectors of $f$ are exactly $\ell_\pi$ applied to the exponent vectors of monomials in $g$ that are disjoint from $\ker(\pi)$. By the linearity of $\ell_\pi$ and the convexity of $P$ and $\New(f)$, we have that $\New(f) = \ell_\pi(P)$, so $P$ is an extension of $\New(f)$. Since $P$ is defined by intersecting $\New(g)$ with additional linear equations, the lemma follows.
\end{proof}

Several partial converses to our Main Lemma also hold. Perhaps the most natural and interesting of these is the following. %

\begin{observation}
Let $R$ be a totally ordered semi-ring. Given any sequence of non-negative %
integral polytopes $P_n \subseteq \R^n$ %
such that the $\poly(n)$-th cycle cover polytope is an extension of $P_n$ along an affine linear map $\ell_n\colon \R^{\poly(n)} \to \R^n$ with integer coefficients of polynomial bit-length, there is a sequence of polynomials $(f_n) \in \cclass{VNP}_R$ such that $\New(f_n) = P_n$ and $f$ is a monotone p-projection of the permanent.
\end{observation}

\begin{proof}
Let $C_m$ denote the $m$-th cycle cover polytope, let $m(n)$ be a polynomial such that $C_{m(n)}$ is an extended formulation 
of $P_n$, and let $b(n)$ be a polynomial upper bound on the bit-length of the coefficients of $\ell_n$. Let $V_m$ denote the vertex set of the cycle cover polytope, \ie, the incidence vectors of cycle covers. Define $f_n$ as
\[
  \sum_{\vec{e} \in V_m} \vec{y}^{\ell_n(\vec{e})}\,, 
\]
where $\vec{y}=(y_1,\dotsc,y_n)$, and the vector notation $\vec{y}^{\vec{e}'}$ is defined as $y_1^{e_1'} y_2^{e_2'} \dotsb y_n^{e_n'}$. Note that $\ell_n$ has only integer coefficients by assumption, and each $\vec{e} \in V_m$ is integral, so the vector $\ell_n(\vec{e})$ is integral, and the above expression is a well-defined Laurent polynomial. %
To see that all the exponents are non-negative, note that by assumption each $\ell_n(\vec{e})$ lies in $P_n$, which is itself a non-negative polytope, so the above expression is in fact a well-defined \emph{polynomial}. %
By construction, every exponent vector of $f_n$ is in $\ell_n(C_{m(n)}) = P_n$. Conversely, every vertex of $P_n$ is an exponent vector of $f_n$, since its coefficient is simply $1 + 1 + \dotsb + 1$, which is never zero in a totally ordered semi-ring. %
(Without noting this, it would be possible that $k$ distinct vertices in $V_m$ would get mapped to the same point under $\ell_n$, %
and then the corresponding monomials $\vec{y}^{\ell_n(\vec{e})}$ might add up to $0$ in $f_n$.) %
Thus $\New(f_n) = P_n$. Furthermore, $f_n$ is a monotone \emph{nonlinear} projection of the permanent using the map $x_{ij} \mapsto \vec{y}^{\ell((0,0,\dotsc,1,\dotsc,0))}$, where the $1$ is in the $(i,j)$ position. Using the universality of the permanent and repeated squaring, this can easily be turned into a monotone \emph{simple} projection of the permanent of size $\poly(m(n), b(n))$. 
\end{proof}

This can be generalized from the cycle cover polytopes and the permanent to arbitrary integral polytopes and the natural associated polynomial (the sum over all monomials whose exponent vectors are vertices of the polytope), but at the price of using ``monomial projections''---in which each variable is replaced by a monomial---rather than simple projections. There ought to be a version of this observation over sufficiently large fields and allowing rational coefficients in $\ell$, using Strassen's division trick~\cite{strassen}, but the only such versions the author could come up with had so many hypotheses as to seem uninteresting.

\section{Applications}

\subsection{Projection lower bounds}

\begin{remark} 
The following theorems hold over any totally ordered semi-ring, including the Boolean and-or semi-ring, the non-negative real numbers under multiplication and addition, and the tropical semi-ring of real numbers under addition and $\min$. To see that this introduces no additional difficulty, note that over any totally ordered semi-ring $R$, the Newton polytope of a polynomial over $R$ is still a polytope in a vector space over the real numbers, so standard results on polytopes and the cited results on extension complexity still apply.
\end{remark}

\begin{theorem}[Main Lemma +
\rothvoss's
TSP polytope lower bound] \label{thm:HC} %
Over any totally ordered semi-ring, the Hamiltonian Cycle polynomial is not a monotone affine p-projection of the permanent; in fact, any monotone affine projection from the permanent to the Hamiltonian Cycle polynomial has blow-up at least $2^{\Omega(n)}$. 
\end{theorem}

We will need the following folklore lemma; as we could not find a proof in the literature, we a sketch its well-known proof here for completeness, and so that it can be verified that the standard proof preserves monotonicity.

\begin{lemma}[Folklore] \label{lem:perm}
If an $n$-variable polynomial is an affine projection of the $m \times m$ permanent, then it is a simple projection of the $N \times N$ permanent for $N = m + (n+1)m^2$. The same holds with ``affine projection'' replaced by ``monotone affine projection'' in both places.
\end{lemma}

\begin{proof}[Proof sketch]
Let $\pi_{ij}(\vec{x})$ be the affine linear function corresponding to the variable $y_{ij}$ of the $m \times m$ permanent, and write $\pi_{ij} = a_0 + a_1 x_1 + \dotsb + a_n x_n$. Let $G$ be the complete directed graph with loops on $m$ vertices and edge weights $y_{ij}$. Replace the edge $(i,j)$ by $n+1$ parallel edges with weights $a_0, a_1 x_1, \dotsb, a_n x_n$. Add a new vertex on each of these parallel edges, splitting each parallel edge into two. For the edge weighted $a_0$, the two edges have weights $1,a_0$, and for the remaining edges the new edges get weights $a_i, x_i$. On each of the added vertices, add a self-loop of weight 1. 
Now consider the permanent of the weighted adjacency matrix of this 
edge-weighted directed graph.  %
It is a simple 
exercise to see that this has the desired effect. 
Note also that if the original affine projection $\pi$ was monotone, then so is the constructed simple projection.
\end{proof}

\begin{proof}[Proof of \expref{Theorem}{thm:HC}] %
By \expref{Lemma}{lem:perm}, it suffices to show the result for simple projections. 
If the Hamiltonian Cycle polynomial were a monotone projection of the permanent, then by the Main Lemma, some face of $\New(\perm)$ would be an extension of $\New(\HC)$. 

The Newton polytope of the permanent is the cycle cover polytope (see \expref{Section}{sec:prelim}). The cycle cover polytope can easily be described by the $m^2$ inequalities
asserting
that all variables $x_{i,j}$ are non-negative, together with the equalities
asserting
that each vertex has in-degree and out-degree exactly 1, namely $\sum_i x_{i,j} = 1$ for all $j$ and $\sum_j x_{i,j} = 1$ for 
all $i$. (It is easy to see that these are necessary; for 
sufficiency, see, \eg, \cite[Theorem~18.1]{schrijver}.) 
Since equalities do not count towards the complexity of a polytope,
we have $\compl(\New(\perm_m)) \leq m^2$.

But the Newton polytope of the $n$-th Hamiltonian Cycle polynomial is exactly the TSP polytope, 
which, by a recent result by
\rothvoss~\cite[Corollary~2]{rothvoss},
requires extension complexity $2^{\Omega(n)}$.
\end{proof}

\begin{theorem}[Main Lemma +
\rothvoss's
perfect matching polytope lower bound] \label{thm:matching} %
Over any totally ordered semi-ring, the perfect matching polynomial (or ``unsigned Pfaffian'') is not a monotone affine p-projection of the permanent; in fact, any monotone affine projection from the permanent to the perfect matching polynomial has blow-up at least $2^{\Omega(n)}$.
\end{theorem}

\begin{proof}
The proof is the same as for the Hamiltonian Cycle polynomial, 
using~\cite[Theorem~1]{rothvoss} by
\rothvoss,
which gives a lower bound of $2^{\Omega(n)}$ on the extension complexity of the perfect matching polytope, which is the Newton polytope of the perfect matching polynomial.
\end{proof}

\begin{theorem}[Main Lemma + Fiorini \emph{et al.}'s cut polytope lower bound] \label{thm:cut} %
Over any totally ordered semi-ring, for any $q$, the $q$-th cut polynomial is not a monotone affine p-projection of the permanent; in fact, any monotone affine projection from the permanent to the $q$-th cut polynomial has blow-up at least $2^{\Omega(n)}$. 
\end{theorem}

\begin{proof}
Use~\cite[Theorem~7]{EF} by Fiorini \emph{et al.}
which says that $\xc(\New(\Cut^2)) \geq 2^{\Omega(n)}$, as $\New(\Cut^2)$ is the cut polytope. The one additional observation we need is that $\New(\Cut^q)$ is just the $(q-1)$-scaled version of $\New(\Cut^2)$, and this rescaling does not affect the extension complexity.
\end{proof}

\subsection{Monotone formula and circuit lower bounds} \label{sec:formula}
As pointed out by an anonymous reviewer, the universality of the permanent 
for formula size %
also holds in the monotone setting, so lower bounds on monotone projections from the permanent imply the same lower bounds on monotone formula size, and therefore quasi-polynomially related lower bounds on monotone circuit size. We assume circuits only have gates of bounded fan-in; with unbounded fan-in, rather than losing a factor of a half in the exponent of the exponent, we lose a factor of a third.

\begin{proposition} \label{prop:formula}
Any polynomial computable by a monotone formula of size $s$ is a monotone projection of $\perm_{s+1}$.
\end{proposition}

\begin{proof}
The proof of the universality of the permanent given in~\cite[Proposition~2.16]{burgisserBook} works \emph{mutatis mutandis} in the monotone setting.
\end{proof}

As a consequence of this, Theorems~\ref{thm:HC}--\ref{thm:cut} are nearly tight, since every monotone polynomial in $n$ variables of $\poly(n)$ degree can be written as a monotone formula of size $2^{O(n \log n)}$ (write it as a sum of monomials).

\begin{corollary}
Over any totally ordered semi-ring, any monotone formula computing the Hamiltonian Cycle polynomial, the perfect matching polynomial, or the $q$-th cut polynomial has size at least $2^{\Omega(n)}$. Consequently, any monotone circuit computing these polynomials has size at least $2^{\Omega(\sqrt{n})}$.
\end{corollary}

For the cut polynomials, we believe this result to be new. For the other polynomials, this provides a new proof of (slightly weaker versions of) previously known lower bounds. Namely, Jerrum and Snir gave a lower bound of $(n-1)((n-2)2^{n-3}+1) = 2^{n+\Omega(\log n)}$ on the monotone circuit size of $\HC$~\cite[Section~4.4]{jerrumSnir}, and a lower bound of $n(2^{n-1}-1)$ on the monotone circuit size of the permanent~\cite[Section~4.3]{jerrumSnir}. As the permanent is a monotone projection of the perfect matching polynomial---namely, restrict the perfect matching polynomial to a bipartite graph, \eg, by setting $x_{ij}=0$ whenever $i$ and $j$ have the same parity---the same lower bound holds for the perfect matching polynomial.

\begin{proof}
The first part follows by combining \expref{Proposition}{prop:formula} with Theorems~\ref{thm:HC}--\ref{thm:cut}. The second part follows from the fact that monotone circuits of size $s$ can be balanced to have size $\poly(s)$ and depth $O(\log^2 s)$ (the proof in~\cite{VSBR} works \emph{mutatis mutandis} in the monotone setting), which can then be converted to monotone formulas of size $s^{O(\log s)} = 2^{O(\log^2 s)}$ by the usual conversion from bounded fan-in circuits to formulas. If there is a monotone circuit of size $s$ computing any of these polynomials, there is thus a monotone formula of size $2^{O(\log^2 s)}$, which must be at least $2^{\Omega(n)}$, so $s \geq 2^{\Omega(n^{1/2})}$.
\end{proof}

\section{Open questions} \label{sec:open}
Despite the common feeling that Razborov's super-polynomial lower bound~\cite{razborov} on monotone Boolean circuits for CLIQUE ``finished off'' monotone Boolean circuit lower bounds, several natural and interesting question remain. For example, does Directed $s$-$t$ Connectivity require monotone Boolean circuits of size $\Omega(n^3)$? (A matching upper bound is given by the Bellman--Ford algorithm.) Is there a monotone Boolean reduction from general perfect matching to bipartite perfect matching? A positive answer to the following question would rule out such monotone (projection) reductions.

\begin{open}
Extend \expref{Theorem}{thm:matching} from formal polynomials over the Boolean semi-ring to Boolean functions.
\end{open}

However, there are even easier questions, intermediate between the Boolean function case and the algebraic case considered in this paper; Jukna~\cite{juknaMonotone} discusses the notion of one polynomial ``counting'' another, which means that they agree on all $\{0,1\}$ inputs. 

\begin{open} \label{open:count}
Prove that no monotone polynomial-size projection of the permanent agrees with the perfect matching polynomial on all $\{0,1\}$ inputs (``counts the perfect matching polynomial''). Similarly, prove that no monotone polynomial-size projection of the permanent counts the Hamiltonian cycle polynomial.
\end{open}

S.\ Jukna points out (personal communication) that projections of the $s$-$t$ connectivity polynomial correspond, even in the Boolean setting, to switching-and-rectifier networks, so the known lower bounds on monotone switching-and-rectifier networks (see, \eg, the survey~\cite{razborovSRN}) imply that the Hamiltonian path polynomial and the permanent are not monotone p-projections of the $s$-$t$ connectivity polynomial, even over the Boolean semi-ring. This helps explain why the only known monotone lower bound on the $s$-$t$ connectivity polynomial that we are aware of~\cite{juknaMonotone} goes by a somewhat roundabout proof: Razborov's lower bound on CLIQUE~\cite{razborov}, followed by Valiant's reduction from the clique polynomial to the Hamiltonian path polynomial~\cite{valiant2}, followed by a standard reduction from Hamiltonian path to counting $s$-$t$ paths. In the course of discussing this, we were led to the following question; although the motivation for the question has since disappeared, it still seems like an interesting question about polytopes, whose answer may require new methods.

\begin{open}[S.\ Jukna, personal communication] \label{open:path}
Is the $m$-th $s$-$t$ path polytope an extension of the $n$-th TSP polytope (or $n$-th cycle cover polytope) with $m \leq \poly(n)$? 
\end{open}

Since the separation problem for the $s$-$t$ path polytope is $\cclass{NP}$-hard (see, \eg, \cite[\S 13.1]{schrijver})---and the cycle cover polytope has low (extension) complexity---answering this question negatively seems to require more subtle understanding of these polytopes than ``simply'' an extended formulation lower bound.

Another example of a natural polytope question with a similar flavor comes from the cut polynomials. In combination with B\"{u}rgisser's results and questions on the cut polynomials~\cite{burgisser} (discussed in \expref{Section}{sec:intro}), we are led to the following question.

\begin{open}
Is the $m$-th cut polytope an extension of the $n$-th TSP polytope, for $m \leq \poly(n)$?
\end{open}

A negative answer would show that $\Cut^q$ is not complete for non-negative polynomials in $\VNP_{\Q}$ under monotone p-projections, though as with the example of the permanent, this is not necessarily an obstacle to being $\VNP$-complete under general p-projections. Yet even the monotone completeness of the cut polynomials remains open. In fact, even more basic questions remain open:

\begin{open} \label{open:poscomplete}
Is every non-negative polynomial in $\cclass{VNP}$ a monotone projection of the Hamiltonian Cycle polynomial? Is there any polynomial that is ``positive $\cclass{VNP}$-complete'' in this sense?
\end{open}

To relate this to the current proofs of $\cclass{VNP}$-completeness of $\HC_n$, we need to draw a distinction. Let $\cclass{VP}_{\R}^{\geq 0}$ denote the polynomial families in $\cclass{VP}_{\R}$ all of whose coefficients are non-negative, and let $\cclass{mVP}_{\R}$ (``monotone $\cclass{VP}$'') denote the class of families of polynomials with polynomially many variables, of polynomial degree, and computable by polynomial-size \emph{monotone} circuits over $\R$. Similarly, define $\cclass{VNP}_{\R}^{\geq 0}$ to be the non-negative polynomials in $\cclass{VNP}_{\R}$, and $\cclass{mVNP}_\R$ to be the function families of the form
\[
  f_n = \sum_{\vec{e} \in \{0,1\}}^{\poly(n)} g_m(\vec{e}, \vec{x})\,,
\]
where $m \leq \poly(n)$ and $(g_m) \in \cclass{mVP}_\R$. 

Valiant's original completeness proof for the Hamiltonian Cycle polynomial~\cite{valiant2} is ``mostly'' monotone: it %
uses polynomial-size formulas for the coefficients of the monomials (coming from the definition of $\cclass{VNP}$), but otherwise is entirely monotone. In other words, the proof shows that $\HC$ is $\cclass{mVNP}$-hard under monotone projections. However, we note that it is %
not clear whether $\HC$ is even \emph{in} $\cclass{mVNP}$! \expref{Question}{open:poscomplete} asks whether $\HC$, or indeed any polynomial, is $\cclass{VNP}^{\geq 0}$-complete under monotone projections; the question of whether there exist polynomials that are $\cclass{mVNP}$-complete under monotone projections also seems potentially interesting.

Finally, we ask about stronger notions of monotone reduction, which seem to require a different kind of proof technique. Recall that a \definedWord{c-reduction} from $f$ to $g$ is a family of polynomial-size %
algebraic circuits for $f$ with oracle gates for $g$. 

\begin{open}
Do the analogues of Theorems~\ref{thm:HC}--\ref{thm:cut} hold for monotone bounded-depth c-reductions in place of affine p-projections? What about weakly-skew or even general monotone c-reductions?
\end{open}

\section{Subsequent developments}
Since the appearance of the preliminary version of this paper~\cite{prelim}, our Main \expref{Lemma}{lem:main} has been used 
by Mahajan and Saurabh~\cite{MS} %
to prove that several other polynomials of combinatorial and complexity-theoretic interest are not
subexponential-size   %
projections of the permanent. 
\begin{enumerate}
\item The $n$-th \definedWord{satisfiability polynomial} over $\F_q$ is a polynomial in $n + 8\binom{n}{3}$ variables denoted $X_1, \dotsc, X_n$ and $\{Y_c : c \in C_n\}$ where $C_n$ denote the set of clauses on 3 literals in $n$ variables. It is defined as
\[
\Sat^q_n(X, Y) = \sum_{a \in \{0,1\}^n} \left(\prod_{i \in [n]} X_i^{q-1} \right) \left( \prod_{c \in C_n : c(a) = 1} Y_c^{q-1} \right)\,.
\]

\item A \definedWord{clow} in an $n$-vertex graph is a closed walk of length exactly $n$, in which the minimum-numbered vertex appears exactly once. The $n$-th \definedWord{clow polynomial} over $\F_q$ is a polynomial in $\binom{n}{2} + n$ variables $X_e$ for each edge $e$ in the complete undirected graph $K_n$ on $n$ vertices and $Y_v$ for each $v \in [n]$. It is defined as
\[
\Clow^q_n(X,Y) = \sum_{w : \text{clow of length $n$}} \left(\prod_{e : \text{edges in $w$}} X_e^{q-1} \right) \left(\prod_{v : \text{distinct vertices in $w$}} Y_v^{q-1} \right),
\]
or more precisely,
\[
\Clow^q_n(X,Y) = \sum_{w = [v_0, \dotsc, v_{n-1}]} \left( \prod_{i \in [n]} X_{(v_{i-1},v_{i \text{ mod } n})}^{q-1} \right) \left(\prod_{v \in \{v_0, \dotsc, v_{n-1}\}} Y_v^{q-1} \right),
\]
where the sum is over clows $w$ and $v_0$ denotes the minimum-numbered vertex in $w$.

\item The \definedWord{clique polynomial} is a polynomial in $\binom{n}{2}$ variables $X_e$:
\[
\Clique_n(X) = \sum_{\substack{T \subseteq \binom{[n]}{2}\\ \text{$T$ is a clique in $K_n$}}} \prod_{e \in T} X_e\,.
\]
\end{enumerate}

\begin{theorem*}[{Mahajan and Saurabh~\cite[Theorems~2 and 6]{MS}}]
Over any totally ordered semi-ring, any monotone affine projection from the permanent to $\Sat^q_n$ or to the clique polynomial requires blow-up at least $2^{\Omega(\sqrt{n})}$. Any monotone affine projection from the permanent to $\Clow^q_n$ requires blow-up at least $2^{\Omega(n)}$.
\end{theorem*}

As in \expref{Section}{sec:formula}, we get the following corollary. Again, we note that the lower bound on the clique polynomial over the Boolean semi-ring only works for the formal clique polynomial (in contrast to Razborov's result~\cite{razborov}, which works for any monotone Boolean circuit computing the CLIQUE \emph{function}).

\begin{corollary} \label{cor:MS}
Over any totally ordered semi-ring, any monotone formula computing $\Clow^q_n$ has size at least $2^{\Omega(n)}$ and any monotone circuit computing $\Clow^q_n$ has size at least $2^{\Omega(\sqrt{n})}$. Any monotone formula computing $\Sat^q_n$ or $\Clique_n$ has size at least $2^{\Omega(\sqrt{n})}$ and any monotone circuit computing these polynomials has size at least $2^{\Omega(n^{1/4})}$.
\end{corollary}

For the clow and satisfiability polynomials, we believe this result to be new. For the clique polynomials, this provides a new proof of a weaker version of %
the exponential monotone circuit lower bound due to Schnorr~\cite{schnorr}.\footnote{Schnorr showed a $\binom{n}{k}-1$ lower bound on the monotone circuit size of the the $k$-th clique polynomial $\Clique^k_n$, the sum over all cliques of size $k$, rather than all cliques. For $k=n/2$, this lower bound is asymptotically equal to $2^n / \sqrt{\pi n/2}$. The $k$-th clique polynomial $\Clique_n^k$ is the degree $k$ homogeneous component of the clique polynomial $\Clique_n$; by homogenization (implicit in Strassen's work, explicit in Valiant~\cite[Lemma~2]{valiantNeg}), any monotone circuit of size $s$ for $\Clique_n$ can be converted into a monotone circuit of size $s(n/2+1)^2$ computing $\Clique^{n/2}_n$. Thus Schnorr's result implies a lower bound of $\Omega(2^n/n^{5/2})$ on the monotone circuit complexity of $\Clique_n$. }

\bigskip\noindent
\textbf{Acknowledgment.} We would like to thank Stasys Jukna for the question that motivated this paper~\cite{juknaQ}, and \url{cstheory.stackexchange.com} for providing a forum for the question. We also thank Stasys for comments on a draft, pointing out the paper by Schnorr~\cite{schnorr}, and interesting discussions leading to Questions~\ref{open:count} and \ref{open:path}. We thank Leslie Valiant for an interesting conversation that led to \expref{Question}{open:poscomplete}. We thank Ketan Mulmuley and Youming Qiao for collaborating on~\cite{GMQ}, which is why the author had Newton polytopes on the mind. We thank an anonymous reviewer for pointing out the monotone universality of the permanent and therefore the implications for monotone formula and circuit size (\expref{Section}{sec:formula} and \expref{Corollary}{cor:MS}). We thank Laci Babai for detailed comments on the journal submission. %

\bibliographystyle{tocplain}   %
\bibliography{v013a018}

\begin{tocauthors}
\begin{tocinfo}[grochow] %
 Joshua A. Grochow\\
 Assistant professor \\
 Departments of Computer Science and Mathematics \\
 University of Colorado, Boulder, CO, USA \\
 jgrochow\tocat{}colorado\tocdot{}edu \\
 \url{http://www.cs.colorado.edu/~jgrochow}
\end{tocinfo}
\end{tocauthors}

\begin{tocaboutauthors}
\begin{tocabout}[grochow]
\textsc{Joshua A. Grochow} is an Assistant Professor in the Departments of Computer Science and Mathematics and the University of Colorado, Boulder. Prior to that he was an Omidyar Postdoctoral Fellow at the \href{http://www.santafe.edu}{Santa Fe Institute}, and a postdoc in the \href{http://www.cs.toronto.edu/theory}{theory group at University of Toronto}. He graduated from The University of Chicago in 2012; his advisors were \href{http://lance.fortnow.com}{Lance Fortnow} and \href{https://www.cs.uchicago.edu/directory/ketan-mulmuley}{Ketan Mulmuley}. In addition to his interests in algebraic and geometric complexity theory, he is also interested in broader issues in complex networks and complex adaptive systems, sometimes referred to in our community as ``the other'' complexity theory.

\end{tocabout}
\end{tocaboutauthors}

\end{document}